\documentclass[12pt,a4paper]{article}
\usepackage[centertags]{amsmath}
\usepackage{amsfonts,amsthm,amssymb}
\usepackage{amssymb}
\usepackage{amsmath}
\usepackage{graphicx}
\usepackage{booktabs}
\usepackage{appendix}
\usepackage[hmargin=2.6cm,vmargin=2.6cm]{geometry}
\usepackage{etoolbox}
\usepackage{tikz}
\usepackage{float}
\usepackage{natbib}
\usepackage[hidelinks]{hyperref}
\usepackage[todonotes={textsize=tiny}]{changes}
\usepackage{comment}
\usepackage{caption, subcaption}

\usepackage{xcolor,framed}
\colorlet{lightgray}{gray!60}
\theoremstyle{definition}
\newcommand{\real}{\mathbb{R}}

\newcommand{\underl}{\underline}

\newcommand{\p}{\mathbf{P}}

\newtheorem{lemma}{Lemma}

\newtheorem{theorem}{Theorem}

\newtheorem{corollary}{Corollary}

\newtheorem{example}{Example}

\newtheorem*{lem:main}{Lemma \ref{lem:main}}
\newcommand{\eps}{\varepsilon}
\newcommand{\be}{\begin{equation}}
\newcommand{\ee}{\end{equation}}

\begin{document}
\title{Informationally Robust Cheap-Talk}
\author{Itai Arieli\thanks{Faculty of Industrial Engineering and Management, Technion, Israel.
	iarieli@technion.ac.il.} \and Ronen Gradwohl\thanks{Department of Economics and Business Administration, Ariel University, Israel. roneng@ariel.ac.il.}
 \and Rann Smorodinsky\thanks{Faculty of Industrial Engineering and Management, Technion, Israel.
	rann@ie.technion.ac.il.}}
\date{}

\maketitle
\begin{abstract}
We study the robustness of cheap-talk equilibria to infinitesimal private information of the receiver in a model with a binary state-space and state-independent sender-preferences. We show that the sender-optimal equilibrium is robust if and only if this equilibrium either reveals no information to the receiver or fully reveals one of the states with positive probability. We then characterize the actions that can be played with positive probability in any robust equilibrium. Finally, we fully characterize the optimal sender-utility under binary receiver's private information, and provide bounds for the optimal sender-utility under general private information.
\end{abstract}

\section{Introduction}
The literature on strategic communication studies how and to what extent an informed party can influence the behavior of others solely through the selective revelation of information. Two of the most-studied models of strategic communication are cheap talk \citep{crawford1982strategic} and Bayesian persuasion \citep{kamenica2011bayesian}.\footnote{Other models include costly signaling \citep{spence1978job} and disclosure of verifiable information \citep{grossman1981informational,milgrom1981good}.} In cheap talk, an informed party (the {\em sender}) strategically chooses what information to reveal to an uniformed party (the {\em receiver}), and the latter then takes the former's preferences into account when interpreting and acting on the revealed information. Bayesian persuasion is similar, except that the sender has the additional power to commit to how she will reveal information prior to observing it herself. Both models have been widely studied in the past decade and have been instrumental in shaping our understanding of how information can be used to affect behavior.

In many models of cheap talk and Bayesian persuasion, the sender-optimal equilibria are well understood---see the  work of \citet{lipnowski2020cheap} for cheap talk and that of \citet{kamenica2011bayesian} for Bayesian persuasion. In this paper we analyze the informational robustness of these equilibria, and determine the extent to which they survive the introduction of an infinitesimal amount of private information to the receiver. We show that, while under Bayesian persuasion all equilibria are informationally robust, under cheap talk some equilibria are much more vulnerable to private information. 

Although the robustness of cheap-talk equilibria to other modeling modifications has been studied before,\footnote{See the literature review below.} robustness to the receiver's private information is particularly important. A main goal in models of strategic communication is to determine the degree to which informational manipulation can influence behavior, and to this end the standard cheap-talk model makes the extreme assumption that only the manipulating party has any information.  But the assumption that the receiver has no other source of information, and relies solely on the sender, is typically not entirely accurate. And if the model's conclusions change dramatically when the assumption is only approximately true---that is, when the receiver does have some other information source, no matter how minor---then the model's predictions are no longer reliable. The main question, then, is whether or not the conclusions of the cheap-talk model are sensitive to minor deviations from the assumption of complete ignorance on the part of the receiver.

In this paper we study this question under state-independent sender utilities \citep[as in][]{lipnowski2020cheap}, and focus on a binary state-space. Our first main result, Theorem \ref{th:fb}, provides a necessary and sufficient condition for the sender-optimal equilibrium to be robust to infinitesimal private information of the receiver. We show that robustness holds if and only if the sender-optimal equilibrium is either trivial or fully reveals one of the states to the receiver with positive probability. This second condition implies that the utility of the sender in the optimal equilibrium without private information of the receiver is equal to the minimal utility from actions taken by the receiver under complete information.

When the sender-optimal equilibrium is not informationally robust, there may nonetheless exist other equilibria that are. In a natural follow-up  to the characterization of Theorem \ref{th:fb}, we proceed to characterize such informationally robust equilibria. First, in Theorem \ref{prop:quadruple}, we fully characterize the subset of tuples of actions that can be played in a robust cheap-talk equilibrium. We also show that a tuple of actions is informationally robust if and only if it is informationally robust when the receiver's (infinitesimal)  private information has binary support. 

Next, we turn to study the optimal utility that the sender can achieve under private information of the receiver. Theorem \ref{theorem:binary} provides a full characterization of the sender's maximal utility under infinitesimal private information with binary support. We show that this utility is a maximum of three expressions: (i) the utility under no information; (ii) the utility under an equilibrium in which one of the sender's messages reveals the state; and (iii) the value of a particular finite two-player zero-sum game that is based on the sender's utility function. In Theorem \ref{theorem:general} we then make use of the above to provide bounds for the sender's maximal utility under general infinitesimal private information of the receiver. Taken together, our results provide a comprehensive characterization of informationally robust cheap-talk communication.

It is worth noting that, unlike some results on equilibrium robustness in the cheap-talk model \citep[see, e.g.,][]{diehl2021non}, our results are not necessarily negative in nature. That is, the degree of vulnerability of the sender's maximal utility to the receiver's private information depends on the problem. As Theorem \ref{th:fb} suggests, in some cases, infinitesimal private information  will have no effect on this maximal utility. Furthermore, as our other theorems suggest, even when this maximal utility is not robust there may exist other non-trivial equilibria that are.

\subsection{Illustrative Examples}
We illustrate our results with an example adapted from \citet{lipnowski2020cheap}.

\begin{example}\label{ex:1} There are two states of the world, $0$ and $1$. A receiver (decision maker, policymaker, he) must decide whether to implement policy $P_0$, which he finds best in state $0$, policy $P_1$, which he finds best in state $1$, or no policy at all ($P_\emptyset$). Suppose the receiver will implement policy $P_\omega$ for $\omega\in\{0,1\}$ if he believes the probability of state $\omega$ is at least $0.6$. The receiver's initial belief is that the two states are equally likely, but there is a sender (expert, she) who knows the true state of the world. The sender has state-independent preferences over the receiver's decision---policy $P_0$ yields utility 3, policy $P_1$ yields utility 4, and policy $P_\emptyset$ yields utility 1, regardless of the realized state of the world. 
\end{example}

Observe that, without communication, the receiver will choose $P_\emptyset$, yielding the sender a utility of 1. Can the sender communicate with the receiver in a way that increases her expected utility in equilibrium?

\begin{figure}
    \centering
    \begin{subfigure}[b]{0.45\textwidth}
    \centering
    \begin{tikzpicture}[baseline = 3cm, scale = 0.3]

        \draw[] (0,1) -- (20,1);
        \draw[->] (0,1) -- (0,16) node[left] {};
        
         \node[] at (8,1) (3) {$|$};
        \node[] at (8,-0.5) (4) {$0.4$};
        \node[] at (10,-0.5) (5) {$\pi$};
        \node[] at (10,1) (5) {$|$};

          \node[] at (12,1) (5) {$|$};
        \node[] at (12,-0.5) (6) {$0.6$};
           \node[] at (20,1) (9) {$|$};
        \node[] at (20,-0.5) (10) {$1$};
         
\node[] at (10,-2) (8) {receiver's belief};
        
        \node[] at (0,2) (1) {$-$};
        \node[] at (-1,2) (2) {$0$};
         \node[] at (0,5) (1) {$-$};
        \node[] at (-1,5) (2) {$1$};        
        \node[] at (0,8) (1) {$-$};
        \node[] at (-1,8) (2) {$2$};        
        \node[] at (0,11) (1) {$-$};
        \node[] at (-1,11) (2) {$3$};
        \node[] at (0,14) (1) {$-$};
        \node[] at (-1,14) (2) {$4$};
        
        \draw[lightgray] (0,11) -- (8,11);
        \draw[lightgray] (8,5) -- (12,5);
        \draw[lightgray] (12,14) -- (20,14);


    \end{tikzpicture}
    \caption{Indirect utility}\label{fig:example1a}
    \end{subfigure}
    \hfill
      \begin{subfigure}[b]{0.45\textwidth}
    \centering
    \begin{tikzpicture}[baseline = 3cm, scale = 0.3]

        \draw[] (0,1) -- (20,1);
        \draw[->] (0,1) -- (0,16) node[left] {};
        
         \node[] at (8,1) (3) {$|$};
        \node[] at (8,-0.5) (4) {$0.4$};
        \node[] at (10,-0.5) (5) {$\pi$};
        \node[] at (10,1) (5) {$|$};

          \node[] at (12,1) (5) {$|$};
        \node[] at (12,-0.5) (6) {$0.6$};
           \node[] at (20,1) (9) {$|$};
        \node[] at (20,-0.5) (10) {$1$};
         
\node[] at (10,-2) (8) {receiver's belief};
        
        \node[] at (0,2) (1) {$-$};
        \node[] at (-1,2) (2) {$0$};
         \node[] at (0,5) (1) {$-$};
        \node[] at (-1,5) (2) {$1$};        
        \node[] at (0,8) (1) {$-$};
        \node[] at (-1,8) (2) {$2$};        
        \node[] at (0,11) (1) {$-$};
        \node[] at (-1,11) (2) {$3$};
        \node[] at (0,14) (1) {$-$};
        \node[] at (-1,14) (2) {$4$};
        
        \draw[lightgray] (0,11) -- (8,11);
        \draw[lightgray] (8,5) -- (12,5);
        \draw[lightgray] (12,14) -- (20,14);

       \draw[loosely dashed] (0,11) -- (12,11);
        \draw[loosely dashed] (12,14) -- (20,14);
        \draw[dotted] (0,11) -- (12,14);
        \draw[dotted] (12,14) -- (20,14);

    \end{tikzpicture}
    \caption{Envelopes}\label{fig:example1-envelopes}
    \end{subfigure}
    \caption{The solid gray line is the sender's indirect utility, the dashed line is its quasiconcave envelope, and the dotted line is its concave envelope.}\label{fig:example1}
\end{figure}
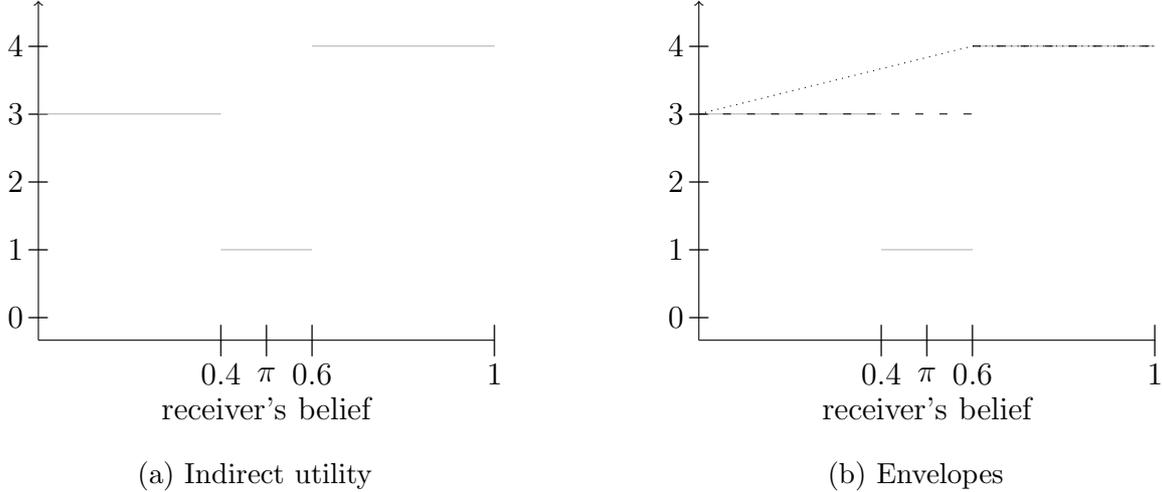

To characterize the gain from communication, consider the sender's indirect utility function---namely, the sender's utility at any  belief of the receiver, assuming that the receiver chooses an optimal action given that belief. The solid gray  line in Figure~\ref{fig:example1a} is the indirect utility in our example, where the x-axis is the receiver's belief that the state is $1$. As noted, at the prior $\pi=0.5$ the indirect utility is 1.

Now, if communication between the sender and the receiver takes the form of Bayesian persuasion---in which the sender commits to how she will communicate with the receiver before observing the realized state---then the maximal utility attainable by the receiver is the concave envelope of the indirect utility \citep{kamenica2011bayesian}. This is illustrated by the dotted line in Figure~\ref{fig:example1-envelopes}, and can be seen to equal about 3.8. If communication takes the form of cheap talk---in which the sender cannot commit, but rather sends a message after observing the state---then the maximal utility attainable by the receiver is the quasiconcave envelope \citep{lipnowski2020cheap}. This is illustrated by the dashed line in Figure~\ref{fig:example1-envelopes}, and can be seen to equal 3.

It is worth examining how the utility of 3 be sustained in an equilibrium with cheap talk. Suppose the sender sends one of two messages, $m_0$ or $m_1$. In state $1$, she always sends the message $m_1$. In state $0$ she mixes between $m_0$ and $m_1$ in such a way that the receiver's posterior upon receiving message $m_1$ is precisely 0.6. Thus, the sender's messages induce beliefs 0 and 0.6.
Furthermore, upon receiving message $m_1$, the receiver is indifferent between policies $P_1$ and $P_\emptyset$; suppose in this case he mixes, and chooses $P_1$ with probability $2/3$. This implies that the sender's expected utility upon sending message $m_1$ is 3, which is equal to her utility on sending message $m_0$. Since the sender is indifferent between the two messages, regardless of the realized state, this strategy profile forms an equilibrium.

Suppose now that the receiver obtains additional information through an information structure $F$. Suppose for simplicity that the information structure is binary and symmetric, yielding realized signals $s_0$ and $s_1$ with probabilities $P(s_0|\omega=0) = P(s_1|\omega=1)=q$, for some $q\in(1/2, 1)$. How does this affect the sender's equilibrium utility?

Observe first that, under Bayesian persuasion, such information can only harm the sender. This is because the sender's commitment is optimal, and if such additional information were beneficial then the sender would have committed to it already. However, the harm to the sender approaches 0 as $q\rightarrow 1/2$ and the information structure becomes uninformative. 

In contrast, under cheap talk, such private information is actually beneficial to the sender. To see this, consider the following equilibrium profile: In state $1$, the sender again always sends message $m_1$. In state $0$, the sender mixes, but this time in such a way that the induced belief on message $m_1$ {\em and} realized signal $s_0$ of the receiver is 0.6. Furthermore, on message $m_1$ and signal realization $s_0$, the receiver mixes between policy $P_1$ and no policy. He mixes in such a way that the sender's indirect utility on belief 0.6 is $x$, where $x$ is such that the expected utility conditional on state $0$ and message $m_1$ is equal to 3: $xq+4(1-q) = 3$. This implies that, in state $0$, the sender is indifferent between messages $m_0$ and $m_1$. Note also that, conditional on state $1$, the sender's utility is $x(1-q)+4q > 3$. This does not violate the sender's incentive constraints, since in state $1$ she always sends message $m_1$ (and so need not be indifferent between the two messages). Note, however, that the sender's unconditional utility is  strictly higher than 3: it is equal to 3 in state $0$, but $x(1-q)+4q > 3$ in state $1$. This is illustrated in Figure~\ref{fig:example1b}.

\begin{figure}
    \centering
    \begin{subfigure}[b]{0.45\textwidth}
    \centering
    \begin{tikzpicture}[baseline = 3cm, scale = 0.3]

        \draw[] (0,1) -- (20,1);
        \draw[->] (0,1) -- (0,16) node[left] {};
        
         \node[] at (8,1) (3) {$|$};
        \node[] at (8,-0.5) (4) {$0.4$};
        \node[] at (10,-0.5) (5) {$\pi$};
        \node[] at (10,1) (5) {$|$};

          \node[] at (12,1) (5) {$|$};
        \node[] at (12,-0.5) (6) {$0.6$};
           \node[] at (20,1) (9) {$|$};
        \node[] at (20,-0.5) (10) {$1$};
         
\node[] at (10,-2) (8) {receiver's belief};
        
        \node[] at (0,2) (1) {$-$};
        \node[] at (-1,2) (2) {$0$};
         \node[] at (0,5) (1) {$-$};
        \node[] at (-1,5) (2) {$1$};        
        \node[] at (0,8) (1) {$-$};
        \node[] at (-1,8) (2) {$2$};        
        \node[] at (0,9.5) (1) {$-$};
        \node[] at (-1,9.5) (2) {$x$};        

        \node[] at (0,11) (1) {$-$};
        \node[] at (-1,11) (2) {$3$};
        \node[] at (0,14) (1) {$-$};
        \node[] at (-1,14) (2) {$4$};
        
    \filldraw [] (0,11) circle (6pt);
        \draw[ lightgray] (0,11) -- (8,11);
        \draw[lightgray] (8,5) -- (12,5);
        \draw[lightgray] (12,14) -- (20,14);
            \filldraw [] (13,14) circle (6pt);
            \draw [->](13,14) .. controls (13.5,15) .. (14,14);
             \draw [->](13,14) .. controls (12.5,15) .. (12,14);
            \draw [] (12,9.5) circle (6pt);



    \end{tikzpicture}
    \caption{Example~\ref{ex:1}}\label{fig:example1b}
    \end{subfigure}
    \hfill
      \begin{subfigure}[b]{0.45\textwidth}
    \centering
    \begin{tikzpicture}[baseline = 3cm, scale = 0.3]

        \draw[] (0,1) -- (20,1);
        \draw[->] (0,1) -- (0,16) node[left] {};
        
         \node[] at (4,1) (1) {$|$};
         \node[] at (4,-0.5) (2) {$0.2$};
         \node[] at (8,1) (3) {$|$};
        \node[] at (8,-0.5) (4) {$0.4$};
        \node[] at (10,-0.5) (5) {$\pi$};
        \node[] at (10,1) (5) {$|$};

          \node[] at (12,1) (5) {$|$};
        \node[] at (12,-0.5) (6) {$0.6$};
           \node[] at (20,1) (9) {$|$};
        \node[] at (20,-0.5) (10) {$1$};
         
\node[] at (10,-2) (8) {receiver's belief};
        
        \node[] at (0,2) (1) {$-$};
        \node[] at (-1,2) (2) {$0$};
         \node[] at (0,5) (1) {$-$};
        \node[] at (-1,5) (2) {$1$};        
        \node[] at (0,8) (1) {$-$};
        \node[] at (-1,8) (2) {$2$};        
        \node[] at (0,11) (1) {$-$};
        \node[] at (-1,11) (2) {$3$};
        \node[] at (0,14) (1) {$-$};
        \node[] at (-1,14) (2) {$4$};
        
       \draw[lightgray] (0,2) -- (4,2);
        \draw[lightgray] (4,11) -- (8,11);
        \draw[lightgray] (8,5) -- (12,5);
        \draw[lightgray] (12,14) -- (20,14);
            \filldraw [] (5,11) circle (6pt);
            \draw [] (4,5) circle (6pt);
            \draw [->](5,11) .. controls (5.5,12) .. (6,11);
             \draw [->](5,11) .. controls (4.5,12) .. (4,11);

\filldraw [] (11,5) circle (6pt);
             \draw [] (12,11) circle (6pt);
             \draw [->](11,5) .. controls (11.5,6) .. (12,5);
              \draw [->](11,5) .. controls (10.5,6) .. (10,5);



    \end{tikzpicture}
    \caption{Example~\ref{ex:2}}\label{fig:example2b}
    \end{subfigure}
    \caption{The solid dots are the induced beliefs, and the hollow circles are the sender's expected utilities at these beliefs (given the receiver's mixing). The arrows illustrate the possible changes in beliefs caused by the receiver's private information.}\label{fig:example}
\end{figure}
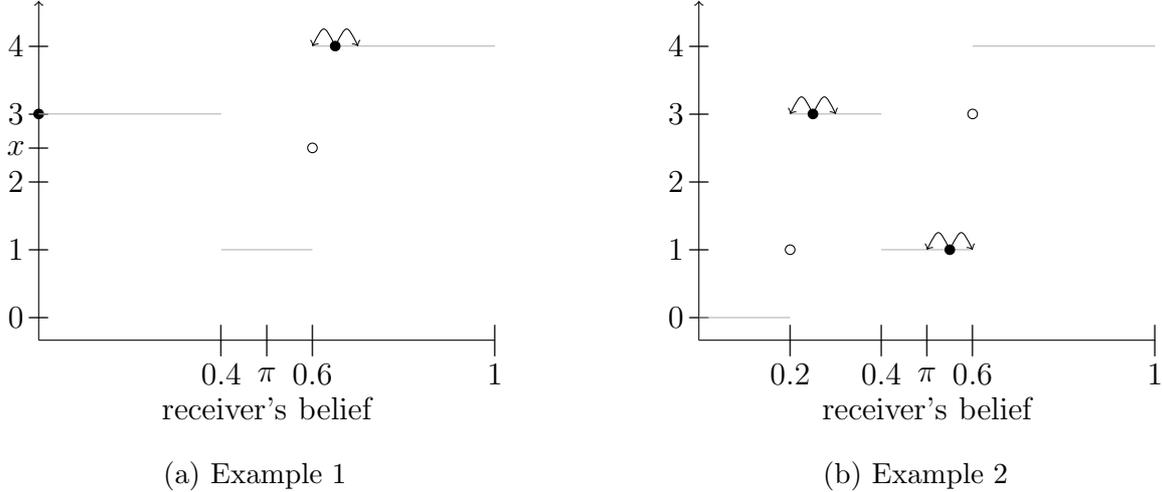

Surprisingly, then, under cheap talk, additional information to the receiver can be beneficial to the sender. Observe, however, that as the receiver's information structure becomes less informative---namely, as $q\rightarrow 1/2$---the sender's utility converges to 3, the same utility as without the receiver's signal. As we show in Theorem~\ref{th:fb}, this last statement is true for {\em any} vanishing information structure of the receiver. In this case we say that the utility of 3 is informationally robust. 

The difference between the no-private-information case and the infinitesimal-private-information case lies in the sender's incentive constraint. In the former,  when the incentive constraint is satisfied and the sender is indifferent between both messages, she is indifferent regardless of the state. In the latter, however, the sender's expected utility depends both on the message sent and on the receiver's signal. In Example~\ref{ex:1}, the expected utilities of the two messages are equal in state 0, and so in that state the sender can mix between them. In state 1, however, the sender is no longer indifferent, since the expected utility under message $m_1$ is higher. However, because here the sender always sends message $m_1$, she need not be indifferent between the messages.

Are all utilities of cheap-talk equilibria informationally robust? In this paper we show that they are not. Consider the following  modification to Example~\ref{ex:1}:
\begin{example}\label{ex:2}
Everything is as in Example~\ref{ex:1}, except that the receiver now has a third potential policy, $R_0$. This policy is best for the receiver only if he believes the state is $0$ with probability at least 0.8. The policy is worst for the sender, and yields her utility 0. 
\end{example}
The indirect utility function of Example~\ref{ex:2}, as well as the concave and quasiconcave envelopes, are illustrated in Figure~\ref{fig:example2}. Observe that the value of the quasiconcave envelope at the prior is still equal to 3; the addition of $R_0$ as an option does not affect the sender's attainable utility under cheap talk.

\begin{figure}
    \centering
    \begin{subfigure}[b]{0.45\textwidth}
    \centering
    \begin{tikzpicture}[baseline = 3cm, scale = 0.3]

        \draw[] (0,1) -- (20,1);
        \draw[->] (0,1) -- (0,16) node[left] {};
        
         \node[] at (4,1) (1) {$|$};
         \node[] at (4,-0.5) (2) {$0.2$};
         \node[] at (8,1) (3) {$|$};
        \node[] at (8,-0.5) (4) {$0.4$};
        \node[] at (10,-0.5) (5) {$\pi$};
        \node[] at (10,1) (5) {$|$};

          \node[] at (12,1) (5) {$|$};
        \node[] at (12,-0.5) (6) {$0.6$};
           \node[] at (20,1) (9) {$|$};
        \node[] at (20,-0.5) (10) {$1$};
         
\node[] at (10,-2) (8) {receiver's belief};
        
        \node[] at (0,2) (1) {$-$};
        \node[] at (-1,2) (2) {$0$};
         \node[] at (0,5) (1) {$-$};
        \node[] at (-1,5) (2) {$1$};        
        \node[] at (0,8) (1) {$-$};
        \node[] at (-1,8) (2) {$2$};        
        \node[] at (0,11) (1) {$-$};
        \node[] at (-1,11) (2) {$3$};
        \node[] at (0,14) (1) {$-$};
        \node[] at (-1,14) (2) {$4$};
        
       \draw[lightgray] (0,2) -- (4,2);
        \draw[lightgray] (4,11) -- (8,11);
        \draw[lightgray] (8,5) -- (12,5);
        \draw[lightgray] (12,14) -- (20,14);



    \end{tikzpicture}
    \caption{Indirect utility}\label{fig:example2a}
    \end{subfigure}
    \hfill
      \begin{subfigure}[b]{0.45\textwidth}
    \centering
    \begin{tikzpicture}[baseline = 3cm, scale = 0.3]

        \draw[] (0,1) -- (20,1);
        \draw[->] (0,1) -- (0,16) node[left] {};
        
         \node[] at (4,1) (1) {$|$};
         \node[] at (4,-0.5) (2) {$0.2$};
         \node[] at (8,1) (3) {$|$};
        \node[] at (8,-0.5) (4) {$0.4$};
        \node[] at (10,-0.5) (5) {$\pi$};
        \node[] at (10,1) (5) {$|$};

          \node[] at (12,1) (5) {$|$};
        \node[] at (12,-0.5) (6) {$0.6$};
           \node[] at (20,1) (9) {$|$};
        \node[] at (20,-0.5) (10) {$1$};
         
\node[] at (10,-2) (8) {receiver's belief};
        
        \node[] at (0,2) (1) {$-$};
        \node[] at (-1,2) (2) {$0$};
         \node[] at (0,5) (1) {$-$};
        \node[] at (-1,5) (2) {$1$};        
        \node[] at (0,8) (1) {$-$};
        \node[] at (-1,8) (2) {$2$};        
        \node[] at (0,11) (1) {$-$};
        \node[] at (-1,11) (2) {$3$};
        \node[] at (0,14) (1) {$-$};
        \node[] at (-1,14) (2) {$4$};
        
       \draw[lightgray] (0,2) -- (4,2);
        \draw[lightgray] (4,11) -- (8,11);
        \draw[lightgray] (8,5) -- (12,5);
        \draw[lightgray] (12,14) -- (20,14);

       \draw[loosely dashed] (0,2) -- (4,2);

      \draw[loosely dashed] (4,11) -- (12,11);
        \draw[loosely dashed] (12,14) -- (20,14);
         \draw[dotted] (0,2) -- (4,11);
        \draw[dotted] (4,11) -- (12,14);
        \draw[dotted] (12,14) -- (20,14);

    \end{tikzpicture}
    \caption{Envelopes}\label{fig:example2-envelopes}
    \end{subfigure}
    \caption{The solid lightgray line is the sender's indirect utility, the dashed line is its quasiconcave envelope, and the dotted line is its concave envelope.}\label{fig:example2}
\end{figure}
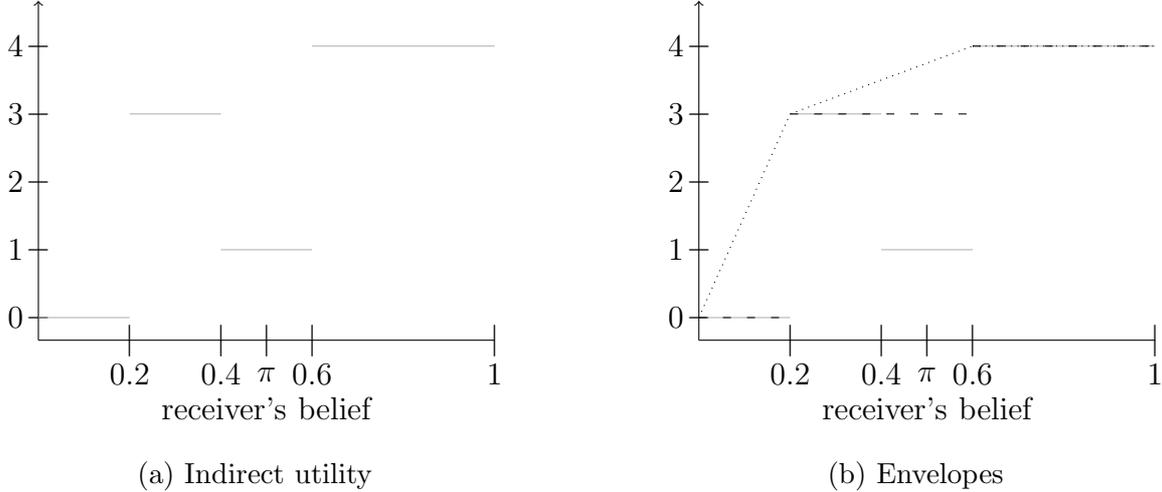

When the receiver has additional private information, however, this equilibrium utility can no longer be sustained.  In fact, Theorem~\ref{th:fb} provides a precise necessary and sufficient condition under which a cheap-talk equilibrium utility is informationally robust: namely, that this utility can be attained at the prior, or that it arises by a pair of sender messages such that at least one message induces either the belief 0 or the belief 1. In Example~\ref{ex:1}, message $m_0$ induces belief 0, and so this utility is informationally robust. In Example~\ref{ex:2}, however, the utility of 3 can only be sustained by beliefs that are interior. Thus, Theorem~\ref{th:fb} implies that the utility of 3 is not attainable: even if the receiver's information is only infinitesimally informative, the sender cannot get the same utility as when the information is completely uninformative.  The intuition is that, because neither induced belief is an endpoint, the sender must mix between both messages in both states, and must thus be indifferent between both messages in {\em both} states. In other words,  in equilibrium both of the sender's incentive constraints require indifference. Theorem~\ref{th:fb} shows that this dual indifference is impossible to attain for the quasiconcave envelope.

However, although the utility of 3 is unattainable, there are other equilibria where the dual indifference is feasible, and so other utilities that are informationally robust. Theorem~\ref{prop:quadruple} characterizes the actions that can be played in an informationally robust equilibrium, and Theorems~\ref{theorem:binary} and~\ref{theorem:general} characterize the utilities that are possible. In Example~\ref{ex:2} with a binary, symmetric information structure of the receiver, a sender-utility of 2 can be attained in an informationally robust equilibrium. This equilibrium is illustrated in Figure~\ref{fig:example2b}.  We defer discussion of the construction and underlying intuitions to Section~\ref{sec:eq-characterization}, as these will be more illuminating after we develop the necessary definitions.

\subsection{Related Literature}
Our paper is most closely related to that of \citet{lipnowski2020cheap}, who take a belief-based approach to study cheap-talk equilibria when the sender has state-independent preferences. A main result of \citet{lipnowski2020cheap} is that the attainable sender utilities equal the quasiconcave envelope of the indirect utility function.\footnote{See also \citet{chakraborty2007comparative,chakraborty2010persuasion}, who study cheap talk in multidimensional settings.} This contrasts with the analogous approach of Bayesian persuasion \citep{kamenica2011bayesian}, in which the attainable utilities are the concave envelope of the indirect utility function. Our baseline setting is identical to that of \citeauthor{lipnowski2020cheap}, except that we restrict our analysis to a binary state-space. We analyze the robustness of equilibria to infinitesimal receiver private-information, and characterize  conditions under which the utilities of the quasiconcave envelope are robustly attainable. When these utilities are not attainable, we characterize the equilibrium actions and utilities.

Our work is related to numerous papers that study cheap-talk models in which the receiver has some private information, including \citet{chen2009communication,chen2012value,de2010cheap,lai2014expert,ishida2016cheap,ishida2019cheap}. Most study variants of the Crawford-Sobel model, in which the state and action spaces are unit intervals, the receiver would like to match the state, and the sender would like to do the same but with some offset \citep{crawford1982strategic}. \citet{ishida2016cheap} is closer to our paper, since they consider a binary-state setting. The main insight in these papers is that, typically, the receiver's private information hinders information transmission, and that the more accurate the receiver's information the coarser the information transmitted in equilibrium. For example, \citet{ishida2016cheap} show that, if the receiver's information is sufficiently accurate, then there is no informative equilibrium even when the players' preferences are close.

There are two main differences between these papers on informed receivers and our work. Most significantly, we focus on infinitesimally informative receiver information, and show that even this has a substantial impact on equilibria. Second, our model departs from that of \citeauthor{crawford1982strategic}, and instead, we consider a sender with state-independent preferences \citep[as in][]{chakraborty2010persuasion,lipnowski2020cheap}.

Other papers look at different forms of communication robustness. \citet{diehl2021non} consider the multidimensional cheap-talk model of \citet{chakraborty2010persuasion}, and show that, when the receiver has Harsanyi-type uncertainty about the sender's utility function, there cannot be any informative communication in equilibrium.
\citet{dilme2022robust} studies a Crawford-Sobel model with a small communication cost, and shows that the only equilibria that are robust to this small cost are highly informative ones.

Finally, our paper has the same motivation as the literature on the robustness of equilibria to a small amount of incomplete information \citep{kajii1997robustness,ui2001robust,morris2005generalized}. This literature is concerned with the robustness of predictions in complete information games to a small amount of uncertainty  about higher-order beliefs. Our paper has an analogous flavor: We are concerned with the robustness of predictions in communication games to a small amount of uncertainty about the receiver's information.

\section{Model}
Consider a binary state-space $\Omega=\{0,1\}$ with a common prior $\pi=\p(\omega=1)$, and a finite set of receiver actions  $A$ with $|A|=\ell$. 
The receiver has a state-dependent utility function $u_R:A\times\Omega\to \real$. He is an expected utility maximizer, and so the set of beliefs for which each action is optimal is a segment. We henceforth assume that action $a_i$ is optimal for beliefs $\lambda\in I_i= [x_{i-1},x_i]$, where $0=x_0< x_1<x_2<\cdots<x_\ell=1$.

The sender has a state-independent utility function $u_S:A\to\real$. We assume that this utility function is {\em generic}, namely, $u_S(a)\neq u_S(b)$ for any two distinct actions $a,b\in A$. 
Define the 
indirect utility of the sender as a set mapping  $v:\Delta(\Omega)\twoheadrightarrow\real$ where, for each belief $\lambda$ about the probability that the state is $\omega=1$, the set $v(\lambda)\subset\real$ consists of all possible utilities for the sender, assuming the receiver best-replies to $\lambda$.

For any belief $\lambda$, let $lr(\lambda)=\log\left(\frac{\lambda }{1-\lambda }\right)$.  Thus, for $i=1,\ldots,\ell$, the intervals $I_i=[x_{i-1},x_i]$ in which action $a_i$ is optimal are translated to $J_i=[y_{i-1},y_i]$, where $y_i=lr(x_i)\in\real\cup\{-\infty,\infty\}$ for every $0\leq i\leq \ell$. For any such interval, let $|J_i|=y_i-y_{i-1}$, and note that $|J_1|=|J_\ell|=\infty$. All other intervals are of finite length.

The receiver obtains private information from an information structure $F=(S,F_0,F_1)$, where $S$ is some measurable space and $F_\omega\in\Delta(S)$ are two probability measures, one for each $\omega\in\Omega$. As usual, identify with each element $s\in S$ the induced posterior belief starting with a prior of $\frac{1}{2}$. Given some prior belief $\pi$, a receiver who receives a signal $s\in S$ will update to a posterior $\pi^s$, where
$lr(\pi^s) \equiv lr(\pi)+lr(s)$. 
We assume that $F_0$ and $F_1$ are mutually absolutely continuous with respect to each other, and thus no signal fully reveals the state.



The sender sends the receiver a message from some finite space $M$. 
In any equilibrium, each message $m\in M$ induces a posterior belief about the likelihood that $\omega=1$. We will henceforth identify the message $m$ with the posterior in $[0,1]$ of $\omega=1$ it generates. Our underlying assumption is that the sender does not observe the realized signal of the receiver. She can thus not base the chosen message on this realized signal, but only on the information structure $F$. Therefore, the sender's message and the receiver's private signals are conditionally independent given the state $\omega$. 

We next define the notion of a cheap-talk equilibrium in our setting with private information of the receiver. A sender's strategy is
a mapping $\sigma:\Omega\to\Delta(M)$ for some message space $M\subseteq[0,1]$. 
A receiver's strategy is a measurable mapping $\rho:M\times S\to\Delta(A)$.  The pair of strategies together with the information structure $F$ and the prior $\pi$ generates a probability distributions $p\in\Delta(\Omega\times M\times S)$. 
Without loss of generality, we assume that $p(\omega=1|m)=m$.

A pair $(\sigma,\rho)$ constitute a cheap-talk equilibrium if 
\begin{enumerate}
    \item The support of $\rho(m,s)$ is contained in the set 
    $\arg\max_{a\in A}p(\omega=1|m,s)u_R(a,1)+p(\omega=0|m,s)u_R(a,0)$.
    \item For every $m\in M$ and $\omega\in\Omega$ it holds that if $p(m|\omega)>0$, then 
    $\int_s u_S(\rho(m,s))\mathrm{d}F_\omega(s)\geq \int_s u_S(\rho(m',s))\mathrm{d}F_\omega(s)$ for every $m'$ that is sent with positive probability. 
\end{enumerate}
We say that the cheap-talk equilibrium is binary if $M$ consists of two elements. We note that by standard considerations the best sender equilibrium can be attained using a binary cheap-talk equilibrium. For this reason, in most of the paper, we will restrict attention to binary equilibria. In addition, we mostly suppress the dependence of the equilibrium on the receiver's strategy and identify a binary cheap-talk equilibrium with the two posteriors 
$m_L<\pi<m_H$ it induces.

Consider the case of a binary equilibrium $m_L<\pi<m_H$. We distinguish two cases, one where 
$m_L,m_H\in(0,1)$ and one where
$\{m_L,m_H\}\cap\{0,1\}\neq\emptyset$. In the first case, the two messages are sent with positive probability in both states $\omega\in\{0,1\}$. This implies that a necessary and sufficient condition for $(\sigma, \rho)$ to be a cheap-talk equilibrium is that for {\em both} states $\omega$, the conditional expected sender's utility given $\omega$ and $m_H$ equals her expected utility given $\omega$ and $m_L$.
Consider now the second case, and suppose $0=m_L<\pi<m_H<1$.\footnote{From our genericity assumption, the case where $m_L=0\text{ and }m_H=1$ can never hold in a cheap talk equilibrium. Also, the case $0<m_L<\pi<m_H=1$ is symmetric to the one under consideration.}  In this case, the two messages are sent with positive probability only in state $\omega=0$, whereas in state $\omega=1$ the message $m_H$ is sent with probability one. Therefore, the equilibrium condition asserts that the conditional expected sender's utility given $\omega=0$ and $m_H$ equals his expected utility given $\omega=0$ and $m_L$. In addition, in state $\omega=1$ the sender's utility from misreporting and sending $m_L$ is not higher than sending the message $m_H$ (but need not be equal to it). This simple observation will play a significant role in our analysis.

Denote the sender's maximal equilibrium utility under information structure $F$  and prior $\pi$ as $v^*_F(\pi)$. If $F$ is uninformative, denote that utility by $v^*_0(\pi)$.
By Lipnowski and Ravid (2020), the optimal value the sender can obtain in any equilibrium with an uninformative $F$ is the quasiconcave closure evaluated at the prior. If the prior is $\pi$, then this optimal value is equal to 
$$v^*_0(\pi) = \min \left\{\max_{\lambda\leq \pi}v(\lambda), \max_{\mu\geq \pi}v(\mu)\right\}.$$
If $v^*_0(\pi)=v(\pi)$, we say that $v^*_0(\pi)$ can be trivially supported.

In this paper, our main question is: What is the maximal 
 guaranteed utility for the sender under a cheap-talk equilibrium, subject to infinitesimal private information of the receiver? To formalize this question, for any $\delta>0$ let $\mathcal{F}_\delta$ be the set of all information structures with support contained in $[\frac 1 2-\delta,\frac 1 2+\delta]$. 
 For any prior $\pi\in[0,1]$ define the \emph{informationally robust equilibrium utility} for the sender at $\pi$ as 
$$\hat v(\pi)=\lim_{\delta\to 0}\inf_{F\in\mathcal{F}_\delta}v_F^*(\pi).$$ 
Our main questions are, does $\hat v(\pi)=v^*_0(\pi)$, or is there a utility discontinuity at the limit? And, in the latter case, what is $\hat v(\pi)$?

\section{Results}\label{sec:results}
We begin this section with our first result, a necessary and sufficient condition under which $\hat v(\pi)=v^*_0(\pi)$. In the subsequent subsections we then characterize informationally robust equilibria and utilities, thereby shedding light on the sender-optimal equilibria in cases where $v^*_0(\pi)$ is not informationally robust.

Throughout, we fix a prior $\pi$ that lies in the interior of some segment $I_j$, namely, $\pi\in \mathrm{int}(I_j)$, for some $j\in [1,\ell]$. At $\pi$, it either holds that 
$v^*_0(\pi)=v(\pi)$ or that $v^*_0(\pi)>v(\pi)$. In the former case,  $v^*_0(\pi)$ is trivially supported. In the latter case, by definition, $v^*_0(\pi)=u_S(a_i)$ for some $i\in[1,\ell]$ with $i\neq j$.  
Given this observation, we now state our first result:
\begin{theorem}\label{th:fb}
For every interval $I_j$ and prior $\pi\in \mathrm{int}(I_j)$ it holds that 
$\hat v(\pi)=v^*_0(\pi)$ if and only if either $v^*_0(\pi)$ is trivially supported  or $ v^*_0(\pi)=u_S(a_i)$ for $i=1$ or for $i=\ell$. 
\end{theorem}
If $v^*_0(\pi)$ is trivially supported then the sender can attain her optimal utility without sending any message. This equilibrium is clearly informationally robust. In the more interesting case in which communication is necessary, 
Theorem~\ref{th:fb} asserts that, in order for the optimal sender-utility to be informationally robust, the sender-optimal equilibrium with no private information must fully reveal one of the states with positive probability. This implies that the sender's optimal utility $v^*_0(\pi)$ must be equal to the smaller of $u_S(a_1)$ and $u_S(a_\ell)$. Note that this is exactly what happens in Example~\ref{ex:1}, where message $m_0$ of the sender reveals that the state is $\omega=0$ to the receiver, leading to a sender-utility of 3.  Observe also that this does not happen in Example~\ref{ex:2}, and there the utility of 3 is not equal to either $u_S(a_1)$ or $u_S(a_\ell)$.

The main intuition underlying Theorem~\ref{th:fb} is the following. The sufficiency of the condition is straightforward, and follows the logic of Example~\ref{ex:1}. 
Suppose, as in that example, that message $m_0=0$, that is, that it fully reveals that the state is $\omega=0$. In that state, the sender mixes between her two messages in such a way that her expected utility given message $m_1$ is equal to $u_S(a_1)$, where the expectation is over the randomness of the receiver's private information as well as the receiver's  mixing over actions. In the proof of Theorem~\ref{th:fb} we show that, under the condition in the theorem, such mixing by the sender and receiver is always possible.

To prove the necessity of the condition in Theorem~\ref{th:fb}, we consider the simple binary, symmetric information structure. Suppose that, without private information, the sender-optimal equilibrium is supported on $0<m_L<\pi<m_H<1$. This implies that the sender must mix between both messages, in both states of the world. 
Her incentive constraints are thus that
$$\int_s u_S(\rho(m_L,s))\mathrm{d}F_\omega(s)= \int_s u_S(\rho(m_H,s))\mathrm{d}F_\omega(s)$$
for {\em each} $\omega\in\{0,1\}$.
In the proof of Theorem~\ref{th:fb} we show that it is impossible to simultaneously satisfy both of these, even for the binary, symmetric information structure, and even when the accuracy of the information provided by that information structure is infinitesimal.

A family of examples in which the sender-optimal utility is informationally robust is where the indirect utility $v:[0,1]\to\real$ changes its trend (from increasing to decreasing or vice versa) at most once. 
If $v$ first increases and then decreases, then the sender-optimal equilibrium is trivially supported. If $v$ first decreases and then increases, then $v^*_0(\pi)=\hat v(\pi)=\min\{u_S(a_1),u_S(a_\ell)\}$.

If the conditions in Theorem~\ref{th:fb} are not satisfied, the sender can no longer attain the maximal, no-private-information utility $v_0^*(\pi)$. However, this does not mean that communication is useless.
In the following subsections we characterize the informationally robust equilibria and utilities, encompassing cases in which 
$\hat v(\pi)<v^*_0(\pi)$, and analyze the extent to which communication can benefit the receiver.

\subsection{Robust Equilibrium Characterization}\label{sec:eq-characterization}
In this section, we characterize the structure of informationally robust equilibria by providing a necessary and sufficient condition for a particular action tuple to be played in such an equilibrium. 

Observe first that, in any cheap-talk equilibrium  with infinitesimal private information, the number of distinct actions that can be played by the receiver in that equilibrium is either one, three, or four. To see this, suppose the sender's messages in equilibrium are $m_L$ and $m_H$, and the prior $\pi\in\mathrm{int}(I_j)$. If $m_L, m_H\in\mathrm{int}(I_j)$ then the equilibrium is trivially supported, and only one action is played in equilibrium.  On the other hand, if $m_L\in I_i$ and $m_H\in I_k$, where $i<k$, then either three or four actions are played: If $m_L$ lies in the interior of $I_i$, then only action $a_i$ can be played in equilibrium following message $m_L$. If $m_L$ lies on the edge of $I_i$ rather then the interior, then the actions played are either $a_i$ and $a_{i+1}$ or $a_{i-1}$ and $a_i$, depending on which edge. The same holds analogously for $m_H$. But note that it cannot be the case that both $m_L$ and $m_H$ lie in the interior of their respective segments, due to our genericity assumption that $u_S(a_i)\neq u_S(a_k)$.

Thus, when an equilibrium is not trivially supported, the action tuples that can be played are either quadruples or triples. However, not all such tuples can be played in an informationally robust equilibrium. In particular, only tuples that belong to one of two sets, $Q$ and $T$, can be played, and we now turn to define these sets.
 
 We first define a set $Q$ of action quadruples. Recall that $\pi \in \mathrm{int}(I_j)$ for some $j$. For $i\leq j-1$ and $k\geq j$ for which $i+1<k$, let $Q$ be the subset of quadruples of actions $\bar a=(a_i,a_{i+1},a_k,a_{k+1})\in A^4$ that satisfy one of the four conditions below:
\begin{enumerate}
    \item $u_S(a_i)<u_S(a_k)<u_S(a_{i+1})<u_S(a_{k+1})$.
\item $u_S(a_{k+1})<u_S(a_{i+1})<u_S(a_k)<u_S(a_i)$.
\item $u_S(a_k)<u_S(a_i)<u_S(a_{k+1})<u_S(a_{i+1})$.
\item $u_S(a_{i+1})<u_S(a_{k+1})<u_S(a_i)<u_S(a_k)$.
\end{enumerate}
The four conditions describe different orderings of sender-utilities. For example, the sender's utilities from the four actions in Example~\ref{ex:2} satisfy condition 1 above. Condition 2 is a mirror image of condition 1. Condition 3 is similar to condition 1, except that the first two actions are swapped with the latter two actions. Finally, condition 4 is a mirror image of condition 3.

Observe that common to these conditions is that the sender-utility from one of the two first actions of $\overline a \in Q$ lies between the utilities from the latter two actions (and vice versa). Thus, these orderings capture only a fraction of all possible orderings---in particular, they exclude orderings where the two first actions yield higher or lower utilities than the latter two actions, or orderings where the first two actions yield the maximal and minimal utilities. Note also that in some cases $Q=\emptyset$. This happens for example when the number of actions $\ell\leq 3$, or when the indirect utility changes its trend at most once.

Now, given a cheap-talk equilibrium, say that it is {\em  supported in the interior} if there are four distinct actions that are played with positive probability by the receiver and that, when ordered according to their indices, these four actions are in $Q$.  

We next define a set $T$ of action triples.
For $1<i<n$, consider the set of all action triples of the following form:
\begin{enumerate}
    \item $(a_1,a_i,a_{i+1})$, where $1<j\leq i$ and $u_S(a_i)<u_S(a_1)<u_S(a_{i+1})$.
    \item $(a_{i-1},a_i,a_\ell)$, where $i\leq j<n$ and $u_S(a_{i+1})<u_S(a_1)<u_S(a_{i})$.
\end{enumerate}
Let $T$ be the set of all such triples. Given a cheap-talk equilibrium, say that it is {\em supported by an endpoint} if three actions are played with positive probability by the receiver, and that, when ordered according to their indices, these three actions lie in $T$.

Finally, say that a tuple of actions $\overline{a}$ is {\em robustly supported} if there exists a $\delta_0>0$ such that, for every $\delta<\delta_0$ and every $F\in \mathcal{F}_\delta$, there exists a  cheap-talk equilibrium where the actions that are played with positive probability are $\overline{a}$. Given these definitions, we can state our second result:
\begin{theorem}\label{prop:quadruple}
A non-singleton tuple of actions $\overline{a}$ is robustly supported if and only if either $\overline{a}\in T$ or $\overline{a}\in Q.$
\end{theorem}

Action triples $T$ essentially cover those equilibria in which one of the sender's messages reveals the state to the receiver. The intuition for why such triples are robustly supported is identical to the sufficiency condition of Theorem~\ref{th:fb}.

Action quadruples $Q$ are a bit more subtle, so let us illustrate this case of Theorem~\ref{prop:quadruple} using Example~\ref{ex:2}. Suppose for simplicity that the information structure is binary and symmetric, and observe that the four policies $(R_0, P_0, P_\emptyset, P_1)$ belong to the set $Q$. How can an informationally robust equilibrium be sustained?

Consider the following strategy profile: the sender sends messages $m_0$ and $m_1$ such that the induced belief on message $m_0$ and signal realization $s_0$ is 0.2, and the induced belief on message $m_1$ and signal realization $s_1$ is 0.6. Furthermore, on belief 0.2 the receiver mixes so that the sender's utility is 1, and on belief 0.6 the receiver mixes so that the sender's utility is 3. This implies that, regardless of the message sent by the sender, if the receiver's signal realization is $s_0$ then the sender's utility is 1, and if the receiver's signal realization is $s_1$ then the sender's utility is 3. Thus, regardless of the state, the sender is indifferent between sending messages $m_0$ and $m_1$. This constitutes an equilibrium, and is illustrated in Figure~\ref{fig:example2b}.  Furthermore,  Theorem~\ref{prop:quadruple} implies that this kind of equilibrium exists for any vanishing information structure of the receiver.




\subsection{Optimal Sender-Utility in Binary Information Structures}\label{sec:binary-info-structure}
In the previous section we characterized the sets of actions that can be played in an informationally robust equilibrium. But which equilibrium leads to the sender's optimal utility, and, moreover, what is the optimal utility?

In this section we fully characterize the sender's optimal utilities for the specific case of binary information structures. In Section~\ref{sec:general-info-structure} below we then use this characterization to provide bounds for the sender's optimal utilities under general information structures.

Recall that the informationally robust equilibrium utility for the sender at $\pi$ is 
$$\hat v(\pi)=\lim_{\delta\to 0}\inf_{F\in\mathcal{F}_\delta}v_F^*(\pi).$$ 
In this section we will limit our attention to receiver's information structures with binary supports, and so we define the \emph{binary informationally robust equilibrium utility} for the sender at $\pi$ as
$$\hat v_b(\pi)=\lim_{\delta\to 0}\inf_{F\in\mathcal{B}_\delta}v_F^*(\pi),$$
where $\mathcal{B}_\delta\subset \mathcal{F}_\delta$ is the set of all information structures with binary support. 

We now provide a full characterization of $\hat{v}_b(\pi)$. 
Consider the case where $Q\neq\emptyset$, and suppose that $Q=\{\bar a_1,\ldots,\bar a_q\}$. Define a matrix $B$ with $q$ rows and $2$ columns as follows:
for every $l\in[q]$ if $\bar a_l$ is of type $1$ or $2$ above, then let $B_{l,1}=u_S(a_k)$ and $B_{l,2}=u_S(a_{i+1})$. If $\bar a_l$ is of type $3$ or $4$ then let $B_{l,1}=u_S(a_{i})$ and $B_{l,2}=u_S(a_{k+1})$. Let $Val(B)$ be the value of the zero-sum game that is defined by $B$, where the row player is the maximizer and the column player is the minimizer. 

Next, say that $u_S(a_1)$ (resp., $u_S(a_\ell)$) is {\em achievable} if there exists $\overline a\in T$ of type 1 (resp., type 2) above. 
  For $i\in\{1,\ell\}$ let $v^*_i=u_S(a_i)$ if $a_i$ is achievable, and $v^*_i=-\infty$ otherwise.
Our characterization of $\hat v_b$ is then the following:
\begin{theorem}\label{theorem:binary}
$\hat{v}_b(\pi)$ equals the maximum of $v^*_1,v^*_\ell$, $Val(B)$, and $v(\pi)=u_S(a_j)$.     
\end{theorem}
 The main idea behind the proof of Theorem \ref{theorem:binary} is the following. Consider Example~\ref{ex:2}, and the equilibrium that 
 is described at the end of Section \ref{sec:eq-characterization} and illustrated in Figure~\ref{fig:example2b}. This equilibrium is supported on the unique quadruple $(a_1,a_2,a_3,a_4)\in Q$, a  quadruple that is of type 1. In this example, the zero-sum game defined by $B$ consists of one row (since there is only one element in $Q$), and its values are $B_{1,1}=u_S(a_k)=u_S(a_3)=u_S(P_\emptyset)=1$ and $B_{1,2}=u_S(a_{i+1})=u_S(a_2)=u_S(P_0)=3$. The value of this game is thus 1. We will now show that this is equal to the sender maximal robust utility.
 
 First, observe that the sender's utility is $u_S(a_3)=1$ if the receiver receives the low signal $s_0$, and $u_S(a_2)=3$ if the receiver receives the high signal $s_1$. This is true regardless of the realized message $m$. Note that $\hat{v}_b(\pi)$ is the infimum over all binary information structures. Since there exists an infinitesimal, binary information structure for which realization $s_0$ is arbitrarily more likely than realization $s_1$, the value of $\hat{v}_b(\pi)$ in this example is 1, the same as $Val(B)$.
 
More generally, if the private information of the receiver is binary and the probability of the high signal is $\alpha$, then the sender's equilibrium utility when playing $\overline a\in Q$ is $\alpha u_S(a_k)+(1-\alpha)u_S(a_{i+1})$ if $\overline{a}$ is of type 1 or 2, and $\alpha u_S(a_i)+(1-\alpha)u_S(a_{k+1})$ if $\overline{a}$ is of type 3 or 4. This implies that, given $\alpha$, the sender's maximal utility across all equilibria that are supported on the interior is $\max_{1\leq l\leq \ell} \alpha B_{l,1}+(1-\alpha)B_{l,2}$. Minimizing this maximum across all possible values of $\alpha$ yields $Val(B)$, the robust sender-utility across all equilibria that are supported on the interior.

Let us illustrate this more general application of Theorem~\ref{theorem:binary} with an example, a further modification of Example~\ref{ex:2}.

\begin{example}\label{ex:3} Everything is as in Example~\ref{ex:2}, except that the receiver now has a fourth potential policy, $R_1$. This policy is best for the receiver only if he believes the state is $1$ with probability at least 0.8. This policy yields the sender utility $x\in(1,3)$.
\end{example}
The indirect utility for Example~\ref{ex:3} is illustrated in Figure~\ref{fig:four-actions}. Observe that, in this example, $Q$ contains two quadruples of actions: $\overline a_1=(a_1,a_2,a_3,a_4)$, which is of type 1, and $\overline a_2=(a_2,a_3,a_4,a_5)$, which is of type 4.

We now use Theorem \ref{theorem:binary} to derive $\hat v_b(\pi).$ 
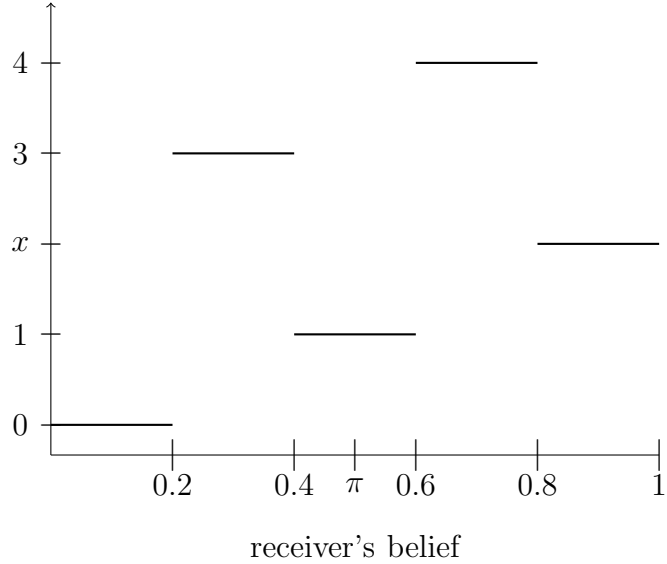
\begin{figure}
    \centering
    \begin{tikzpicture}[baseline = 3cm, scale = 0.4]

        \draw[] (0,1) -- (20,1);
        \draw[->] (0,1) -- (0,16) node[left] {};
        
        \node[] at (4,1) (1) {$|$};
        \node[] at (4,0) (2) {$0.2$};
         \node[] at (8,1) (3) {$|$};
        \node[] at (8,0) (4) {$0.4$};
\node[] at (10,0) (5) {$\pi$};
\node[] at (10,1) (5) {$|$};

          \node[] at (12,1) (5) {$|$};
        \node[] at (12,0) (6) {$0.6$};
          \node[] at (16,1) (7) {$|$};
        \node[] at (16,0) (8) {$0.8$};
        \node[] at (20,1) (9) {$|$};
        \node[] at (20,0) (10) {$1$};
         
\node[] at (10,-2) (8) {receiver's belief};
        
        \node[] at (0,2) (1) {$-$};
        \node[] at (-1,2) (2) {$0$};
         \node[] at (0,5) (1) {$-$};
        \node[] at (-1,5) (2) {$1$};        
        \node[] at (0,8) (1) {$-$};
        \node[] at (-1,8) (2) {$x$};        
        \node[] at (0,11) (1) {$-$};
        \node[] at (-1,11) (2) {$3$};
        \node[] at (0,14) (1) {$-$};
        \node[] at (-1,14) (2) {$4$};

       
        \draw[thick] (0,2) -- (4,2);
        \draw[thick] (4,11) -- (8,11);
        \draw[thick] (8,5) -- (12,5);
        \draw[thick] (12,14) -- (16,14);
        \draw[thick] (16,8) -- (20,8);

    \end{tikzpicture}
    \caption{Example~\ref{ex:3}}\label{fig:four-actions}
\end{figure}
First, observe that the matrix $B$ is
\begin{table}[H]
\begin{center}
   
\begin{tabular}{lllll}
\cline{1-2}
\multicolumn{1}{|l|}{$1$} & \multicolumn{1}{l|}{$3$} &  &  &  \\ \cline{1-2}
\multicolumn{1}{|l|}{$3$} & \multicolumn{1}{l|}{$x$} &  &  &  \\
 \cline{1-2}
\end{tabular}
\end{center}
\end{table}
The value of $B$ is $Val(B)=\frac{9-x}{5-x}.$
In addition, we have that $v^*_1=-\infty$, $v^*_5=x$, and $v(\pi)=u_S(a_3)=1$. Since $Val(B)>x>1$, Theorem \ref{theorem:binary} implies that $\hat v_b(\pi)=\frac{9-x}{5-x}.$

\subsection{Optimal Sender-Utility in General Information Structures}\label{sec:general-info-structure}
In the previous section we characterized the maximal sender-utility that is robustly attainable under binary information structures of the receiver. In this section we use that result in order to provide bounds on the maximal sender-utility under general information structures. To this end, consider the matrix $B$ from Section~\ref{sec:binary-info-structure}. Let $V(B)$ be the {\em pure} min-max value of the zero-sum game defined by $B$, where the maximizing row player is restricted to pure strategies. The following theorem bounds $\hat{v}(\pi)$:
\begin{theorem}\label{theorem:general}
$\hat{v}(\pi)$ is bounded from above by $\hat{v}_b(\pi)$ and bounded from below by the maximum of $v^*_1,v^*_\ell$, $V(B)$, and $v(\pi)$.     
\end{theorem}
We note that the upper bound in Theorem~\ref{theorem:general} is trivial, since $\mathcal{B}_\delta\subset \mathcal{F}_\delta$. The lower bound follows a similar logic as that of Theorem~\ref{theorem:binary}. However, there is an important difference. Under the binary information structures of Theorem~\ref{theorem:binary}, each such structure yields realization $s_0$ with some probability $\alpha$ and realization $s_1$ with probability $1-\alpha$. For any message $m$ of the sender, with $m\in I_i$ for some interval, the information structure then leads either to action $a_i$ or to pairs $(a_{i-1}, a_i)$ or $(a_i, a_{i+i})$, where the first element of the pair is played with probability $\alpha$ and the latter with probability $1-\alpha$. This distribution over actions, $\alpha$ and $1-\alpha$, is the {\em same} distribution for any message $m$ of the sender. However, when the information structure is no longer binary, then different messages could potentially lead to different distributions over actions. In this case, the value of the zero-sum game may no longer capture the precise utility attainable.

Nonetheless, we can still use Theorem~\ref{theorem:general} to derive bounds on the optimal informationally robust sender-utility under general information structures. In Example~\ref{ex:3}, for instance, note that, since $1<x<3$, the pure min-max value of the matrix $B$ (defined above) is $x$. Hence, $V(B)=x$. Theorem \ref{theorem:general} thus implies that $x\leq\hat v(\pi)\leq \frac{9-x}{5-x}$.

\section{Conclusion}\label{sec:conclusion}
In this paper we studied the informational robustness of cheap-talk equilibria. We derived a necessary and sufficient condition under which the utility at the quasiconcave closure is robust, and characterized the structure of equilibria and the maximal robust sender-utilities when the quasiconcave closure is not robust. One immediate question still left open by our work is to provide a full characterization of the optimal robust sender-utilities under general information structures. Our conjecture is that, like with binary information structure, the optimal robust sender-utility is captured by the value of the zero-sum game $B$. An additional open question is to generalize our results to a larger state-space.

\bibliography{info-robust-CT}

\newpage
\appendix
\begin{center}
\begin{Large}
\textbf{Appendix}
\end{Large}
\end{center}
\section{Proofs}
\subsection{Lemmas}
We begin with some lemmas that will be useful in our proofs.
\begin{lemma}\label{lem:ivt}
Let $\mu\in\Delta(\real)$ be any probability measure with bounded support. For every $y\in\real$, let $\psi_y=y+\mu$ be the $y$ translation of $\mu$. For any $z,c,d\in\real$ define a  correspondence
$H:\real\twoheadrightarrow [0,1]$ as follows:
$$H(y)=\{c\psi_y((-\infty,z))+d\psi_y((z,\infty))+ \psi_y(\{z\})(r c+(1-r)d)|r\in[0,1]\}.$$ 
Then the range of $H$ is $[c,d]$ if $c<d$ and $[d,c]$ if $d<c$.
\end{lemma}
\begin{proof}[Proof of Lemma \ref{lem:ivt}]
Assume without loss of generality that $c<d$. Let $H^*(y)=\psi_y((-\infty,z))c+\psi_y([z,\infty))d$, and $H_*(y)=\psi_y((-\infty,z])c+\psi_y((z,\infty))d$.
It is enough to show for every $y$ that if 
$\underl h=\lim_{x\to y^+}H^*(x)<\overline h=\lim_{x\to y^-}H^*(x)$, then $H(y)=[\underline h,\overline h].$ To see this note that $\psi_x$ weakly converges to $\psi_y$ as $x$ approaches $y$. Therefore, by the Portmanteau Theorem \citep[Theorem 2.1 in][]{billingsley2013convergence} it holds that $\lim_{x\to y^+}\psi_x([z,\infty))\leq \psi_x([z,\infty))$. Therefore, $H^*(y)\geq \lim_{x\to y^-}H^*(x).$ But since $c<d$ and since $H^*$ is monotonically non-decreasing we have that $\lim_{x\to y^-}H^*(x)=H(y)=\overline h$.  
Similarly, $\lim_{x\to y^+}\psi_x((-\infty,z])\leq \psi_y((-\infty,z])$. Therefore, $H_*(y)\leq \lim_{x\to y^+}H_*(x)$. The monotonicity of $H_*$ then implies that $H_*(y)=\lim_{x\to y^+}H_*(x)=\underl h$, as desired.

\end{proof}

In the following, denote by $\delta(c)$ the Dirac delta function at $c$, and by $\beta \delta(c) + (1-\beta)\delta(c')$ the distribution over distributions $\delta(c)$ and $\delta(c')$, where the former has probability $\beta$.
\begin{lemma}\label{lem:ind}
Let $\pi$ be the prior. Let $F$ be private information for the receiver with binary signals $S=\{s_0,s_1\}$ that induce the posterior distribution $\beta\delta(a)+(1-\beta)\delta(b)$ for some $\beta\in(0,1)$ and $0<b<a$ such that $\beta a+(1-\beta)b=\frac{1}{2}$. Let $(m_l,m_h)$ be the sender's messages in a cheap-talk equilibrium with two messages, where $0<m_l<\pi<m_h<1$. For $k\in\{0,1\}$ let $r_{l,k}$ and $r_{h,k}$ be the expected utilities of the sender conditional on the private signal $s_k$ of the receiver, given the messages 
$m_l$ and $m_h$, respectively. Then $r_{h,k}=r_{l,k}$ for each $k\in\{0,1\}$ .
\end{lemma}
\begin{proof}[Proof of Lemma \ref{lem:ind}]
Assume that $s_1$ induces the high posterior $a$ and $s_0$ the low posterior $b$.  
Note that since $m_l$ and $m_h$ induce a cheap-talk equilibrium and since $\{m_l,m_h\}\cap\{0,1\}=\emptyset$, both signals are sent with positive probability conditional on the two states $\omega\in\{0,1\}$. Therefore, the sender must be indifferent between the two messages $m_l,m_h$ condition on any state $\omega\in\{0,1\}$. It follows from Bayes rule that the posterior distribution of the receiver's private belief conditional on state $\omega=1$ is $F_1=2a\beta\delta(a)+(1-2a\beta)\delta(b)$, and the posterior distribution conditional on state $\omega=0$ is $F_0=2(1-a)\beta\delta(a)+(1-2(1-a)\beta)\delta(b).$
The indifference conditional on state $\omega=1$ implies the  inequality
\begin{align}\label{indif1}
2a\beta r_{h,1}+(1-2a\beta)r_{h,0}=2a\beta r_{l,1}+(1-2a\beta)r_{l,0}.
\end{align}
Similarly, the indifference conditional on state $\omega=0$ implies
\begin{align}\label{indif0}
2(1-a)\beta r_{h,1}+(1-2(1-a)\beta)r_{h,0}=2(1-a)\beta r_{l,1}+(1-2(1-a)\beta) r_{l,0}.
\end{align}
Since $2a\beta\neq 2(1-a)\beta$ the two inequalities can hold simultaneously if and only if $r_{h,1}=r_{l,1}$ and $r_{h,0}=r_{l,0}$, as claimed.
\end{proof}
The following lemma will be used both for Theorem \ref{th:fb} and Theorem \ref{prop:quadruple}.
\begin{lemma}\label{lem:endpointExistence}
Every $\overline{a}\in T$ is robustly supported.   
\end{lemma}
\begin{proof}[Proof of Lemma~\ref{lem:endpointExistence}]
Let $F\in\Delta([0,1])$ be an information structure supported on $[\frac 1 2-\delta,\frac 1 2+\delta]$ such that $2\epsilon:=\log\Big(\frac{\frac{1}{2}+\delta}{\frac{1}{2}-\delta} \Big)<|J_k|$ 
for every $1\leq k\leq \ell$. Recall that $F_\omega$ be its conditional distribution given state $\omega\in\{0,1\}$ and $G_\omega\in\Delta(\real)$ is the corresponding log-likelihood distribution to $F_\omega$. Note that $G_\omega$ is supported on $[-\epsilon,\epsilon]$. 

Without loss of generality assume that $\overline a=(a_1,a_i,a_{i+1})$ for some $1<i<\ell$. We will show there exists a binary cheap talk with a support  $\overline a$.

Let $\lambda\in(0,1)$. Consider a decision problem for the receiver where the prior is $\lambda$ and the private signal is drawn according to $F$. 
Assume that $y=lr(\lambda)\in[y_i-\epsilon,y_i+\epsilon]$. 
Note that the log-likelihood of the posterior belief for the receiver is distributed according to 
the measure $y+\lambda G_1+(1-\lambda)G_0$. 
By construction, for $y\in[y_i-\epsilon,y_i+\epsilon]$ the support of the log-likelihood distribution lies in $J_i\cup J_{i+1}$. Therefore, the posterior belief of the receiver lies in $I_i\cup I_{i+1}.$ Let $\psi_y=y+G_0$ be the conditional log-likelihood distribution of posteriors given state $\omega=0$ as a function of $y$.
Assume that in case of indifference between actions $a_i$ and $a_{i+1}$ (which happens for the posterior $x_{i}$) the receiver plays action $a_i$ with probability $r$ and action $a_{i+1}$ with probability $1-r$.
 
 The conditional expected utility for the sender, given the state $\omega=0$, is 
$$u_S(a_{i})\psi_y(-\infty,y_i)+u_S(a_{i+1})\psi_y(y_i,\infty)+[ru_S(a_i)+(1-r)u_S(a_{i+1})]\psi_y(y_i).$$

We note that the set of all possible expected utilities given the state $\omega=0$, as a function of $y$ and when the receiver plays rationally in the above decision problem,  is given exactly by the set of values $H(y)$ that is defined in Lemma \ref{lem:ivt}. 
Since $u_S(a_i)<u_S(a_1)<u_S(a_{i+1})$ we can rely on Lemma \ref{lem:ivt} and have a point $y'\in[y_k-\epsilon,y_k+\epsilon]$ and a value $r$ such that 
$$u_S(a_1)=u_S(a_{i})\psi_y(-\infty,y_i)+u_S(a_{i+1})\psi_y(y_i,\infty)+[ru_S(a_i)+(1-r)u_S(a_{i+1})]\psi_y(y_i).$$
Let $\lambda'=lr^{-1}(y').$

We consider the following cheap-talk equilibrium.
The sender sends two messages $\{m_l,m_h\}$. Message $m_l$ reveals the state $\omega=0$ and corresponds to a posterior $0$. Message $m_h$ corresponds to the posterior $\lambda'$. We assume that if, after observing his private information, the receiver is indifferent between $a_{i}$ and $a_{i+1}$, then he plays action $a_i$ with probability $r$. We note that conditional on state $\omega=1$, the message $m_h$ is sent with probability one. Conditional on state $\omega=0$ both messages are sent with positive probability. Therefore, in order to show that the above signals constitute a cheap-talk equilibrium we need to show that (i) in state $\omega=0$ the sender is indifferent between the two messages and (ii) in state $\omega=1$ sending $m_h$ is better than sending $m_l$. Condition (i) follows directly from the construction. Condition (ii) Follows since $F_1$ first order stochastically dominates $F_0$ and Therefore, $F_1((x_i,1])$ is higher than $F_0((x_i,1])$. This implies that the utility of $m_h$ conditional on state $\omega=1$ is larger than $u_S(a_1)$. Therefore, signal $m_h$ yields a higher utility to the sender than the message $m_l$ conditional on $\omega=1$. This completes the converse direction.    
\end{proof}

The following lemma is complementary to Lemma \ref{lem:endpointExistence}.
\begin{lemma}\label{lem:inermExistence}
  Every $\overline a\in Q$ is robustly supported.  
\end{lemma}
  
\begin{proof}[Proof of Lemma~\ref{lem:inermExistence}] 
Choose $\delta_0>0$ small enough so that $|J_i|>2\log\Big(\frac{\frac 1 2+\delta_0}{\frac 1 2-\delta_0}\Big)$ for every $i$.
Let $F_\omega$ be the conditional distribution of $F$ given state $\omega\in\{0,1\}$, and let $G_\omega\in\Delta(\real)$ be the corresponding log-likelihood distribution of $F_\omega$. Let  $\inf \mathrm{supp}(G_\omega) = -\eta$ and $\sup \mathrm{supp}(G_\omega) = \eps$. Note that, by construction, $|J_i|>\eta + \eps$. Recall the notation that $J_i=[y_{i-1}, y_i]$.

We begin by considering the cases in which the quadruple $(a_i,a_{i+1},a_k,a_{k+1})\in Q$ is of type 3 or 1. In the former case, let $x_1 = y_{i-1}$, $x_2=y_i$, $x_3 = y_{i+1}$,  $x_4 = y_{k-1}$, $x_5=y_k$, and $x_6 = y_{k+1}$. In the latter case, let $x_1 = y_{k-1}$, $x_2=y_k$, $x_3 = y_{k+1}$,  $x_4 = y_{i-1}$, $x_5=y_i$, and $x_6 = y_{i+1}$.

For both cases, let $U(x)|r$ be the corresponding utility of the sender when the receiver's belief (in terms of log-likelihood) is $x$, and when the receiver randomizes and plays the higher action (the one on the right) with probability $r$ when indifferent. Lemma~\ref{lem:two-equalities} now implies that there exist $v_b<v_w$ and $r_b, r_w\in [0,1]$ such that
$|v_b-x_2|\leq\max\{\eps,\eta\}$, $|v_w-x_5|\leq\max\{\eps,\eta\}$, and the following two qualities hold:
$$E_{Z\sim v_b+G_0}[U(Z)|r_b]=E_{Z\sim v_w+G_0}[U(Z)|r_w] ~~\mbox{ and }~~ E_{Z\sim v_b+G_1}[U(Z)|r_b]=E_{Z\sim v_w+G_1}E[U(Z)|r_w].$$
When the sender sends messages leading to log-likelihood ratios $v_b$ and $v_w$, then, both indifference conditions hold. Thus, this forms a cheap-talk equilibrium with private information $F$. Furthermore, in this equilibrium only the actions $(a_i,a_{i+1},a_k,a_{k+1})$ are played, as required.

Finally, if the quadruple $(a_i,a_{i+1},a_k,a_{k+1})\in Q$ is of type 4 or 2, 
then swap the labeling of the states $\omega\in\{0,1\}$. By symmetry, the proof above goes through in an identical fashion.
\end{proof}

\subsection{Proof of Theorem \ref{th:fb}}

 We next proceed to the proof of Theorem \ref{th:fb}.
\begin{proof}[\textbf{Proof of Theorem \ref{th:fb}}]
We first show the converse direction. Clearly, $v^*_0(\pi)=u_S(a_j)$ implies that
$\hat{v}(\pi)=u_S(a_j)$. Assume that $v^*_0(\pi)=u_S(a_i)\neq u_S(a_j)$ for $i\in\{1,\ell\}$.  Without loss of generality assume that $v^*_0(\pi)=u_S(a_1)$. Since $v^*_0(\pi) = \min \left\{\max_{\lambda\leq \pi}v(\lambda), \max_{\mu\geq \pi}v(\mu)\right\}$,
there exists an interval $i\geq j$ such that $u_S(a_{i})<u_S(a_1)<u_S(a_{i+1})$. Therefore,
$(a_1,a_i,a_{i+1})\in T$, and, by Lemma \ref{lem:endpointExistence}, there exists $\delta_0$ such that for every $\delta<\delta_0$ and any $F\in\mathcal{F}_\delta$ there exists a cheap-talk equilibrium supported on   
$(a_1,a_i,a_{i+1})$. It readily follows from the construction that the sender's utility in any equilibrium supported on $(a_1,a_i,a_{i+1})$ approaches $u_S(a_1)$ as $\delta$ approaches zero. 

We next show that if $v^*_0(\pi)>u_S(a_j)$ and $v^*_0(\pi)\neq u_S(a_i)$ for $i=1,\ell$, then $\hat v(\pi)<v^*_0(\pi).$ Assume that $v^*_0(\pi)=u_S(a_i)$.  Without loss of generality, we can assume that 
$v^*_0(\pi)=\max_{k> j} u_S(a_k)$.
Assume by way of contradiction that $\hat v(\pi)= v^*_0(\pi).$

For every $n$ let $F^n$ be a binary information structure that generates the posterior distribution $\frac{1}{2}\delta\left(\frac{1}{2}-\frac{1}{n}\right)+\frac{1}{2}\delta\left(\frac{1}{2}+\frac{1}{n}\right)$.
We identify the low posterior $\frac{1}{2}-\frac{1}{n}$ with the signal $s^n_0$ and the high posterior with the signals $s^n_1$.
By the contradiction assumption, we can find a sequence of equilibria that yield a utility of $c^n$ to the sender such that $\lim_n c^n=u_S(a_i)$. Without loss of generality, we can further assume that each of the equilibrium is induced by two messages $\{m_l^n,m^n_h\}$ such that $m^n_h>\pi$ and $m^n_l<\pi$. 

We first note that since signals become uninformative, we must have that the total variation distance of $F^n_1$ and $F^n_0$ approaches zero with $n$.  We contend that the conditional expected utility to the sender given the high message $m^n_h$ approach $v^*_0(\pi)=u_S(a_i)$ as $n$ goes to infinity. For if not, then the conditional utility must lie $\eta>0$ below $u_S(a_i)$ for some $\eta>0$ and for some subsequence $\{m_l^{n_k},m^{n_k}_h\}_k$. This implies that the conditional utility given the low signal $m_l^{n_k}$ lies $\eta>0$ above  $u_S(a_i)$ for some $\eta>0$. Therefore, sending the high message yields a strictly larger utility to the sender than sending the low message conditional on both states $\omega$. We thus have a contradiction to the equilibrium assumption.

Since the conditional expectation to the sender given the high message $m^n_h$ approaches $u_S(a_i)$, we must have that $m^n_h\in[x_{i-1},x_i]$ for all sufficiently large $n$. In addition, the conditional expectation to the sender given the low signal also approaches $u_S(a_i)$ as $n$ goes to zero. Since $u_S(a_0)\neq u_S(a_i)$ it follows that $\{m_l^n,m^n_h\}\cap\{0,1\}=\emptyset$ for all sufficiently large $n$. 
As a result, conditional on both states $\omega\in\{0,1\}$ the sender's expected utility from sending both signals is equal. 

As in Lemma \ref{lem:ind}, for $k\in \{1,0\}$, let $r^n_{h,k}$ and $r^n_{l,k}$ be the expected utility to the sender conditional on the private signal $s^n_k$ of the receiver given the messages 
$m_l^n,m^n_h$ respectively.
It follows from Lemma \ref{lem:ind} that $r^n_{h,1}=r^n_{l,1}$ and $r^n_{h,0}=r^n_{l,0}$. Since the conditional expectation to the sender given the high message $m^n_h$ approaches $u_S(a_i)$, we must have that both $r^n_{h,0}$ and $r^n_{h,1}$ approaches $u_S(a_i)$ as $n$ goes to infinity. Since there exists a value $k\neq i$ such that, for all sufficiently large $n$, either $r^n_{l,0}=u_S(a_k)$ or $r^n_{l,1}=u_S(a_k)$, it must hold that $u_S(a_k)=u_A(a_i)$. 
This, however, stands in contradiction to our genericity assumption. 
\end{proof}

\subsection{Proof of Theorem~\ref{prop:quadruple}}
We begin with the following lemma.
\begin{lemma}\label{lem:utility}
Let $F$ be a private binary information structure for the receiver that is supported on $[\frac{1}{2}-\delta,\frac{1}{2}+\delta]$ for some $\delta>0$ such that $2\log(\frac{\frac{1}{2}+\delta}{\frac{1}{2}-\delta})<|J_i|$ for every $i\in[l]$. If $0<m_l<\pi<m_h<1$ defines a non-trivial binary cheap-talk equilibrium, then the equilibrium is supported in the interior. Moreover, in this case it holds that if the equilibrium actions are of types 1 or 2 from $Q$ then $r_{h,0}=r_{l,0}=u_S(a_k) $ and $r_{h,1}=r_{l,1}=u_S(a_{i+1})$, and if they are of types 3 or 4 then $r_{h,0}=r_{l,0}=u_S(a_i)$ and $r_{h,1}=r_{l,1}=u_S(a_{k+1})$.  
\end{lemma}
\begin{proof}[Proof of Lemma~\ref{lem:utility}]
The condition over $F$ guarantees that for any posterior $m\in[0,1]$ the agents' private beliefs are supported on at most two distinct intervals $[x_{i-1},x_i]$. This means that no more than $4$ actions are played with a positive probability at the equilibrium. 

 We claim that $r_{h,0}\neq r_{h,1}$. Assume by way of contradiction that $r_{h,0}= r_{h,1}$. 
Then, since at most two actions are played conditional on any message, it follows from the genericity assumption that $r_{h,0}= r_{h,1}=u_S(a_i)$ for some $i$. By Lemma \ref{lem:ind} this implies that 
$r_{l,0}= r_{l,1}=u_S(a_i)$ for all $n$. This is possible only if $a_i=a_j$ which means that the equilibrium is trivial. 
Similarly, we must also have that $r_{l,0}\neq r_{l,1}$. Therefore, conditional on $m_l$ two consecutive actions $a_i,a_{i+1}$ are played with positive probability, and conditional on $
m_h$ two consecutive actions $a_k,a_{k+1}$ are played with positive probability. 
Assume first that $u_S(a_i)<u_S(a_{i+1}),$ we show that either condition 1. or condition 3. must hold.

We first note that conditional on $m_l$ and the low signal $s_0$ action $a_i$ is played with positive probability for otherwise since the receiver's posterior probability given $m_l$ and $s_1$ is strictly higher, action $a_i$ would have been played with probability zero. This is a contradiction to the fact four distinct actions are played with positive probability. Similarly, conditional on $m_l$ and the high signal $s_1$ action $a_{i+1}$ is played with positive probability.

We distinguish two cases. Assume first that action $a_{i+1}$ is also played with positive probability given $m_l$ and the low signal $s_0$. In this case, since both actions are played with positive probability, the receiver must be indifferent hence his posterior given $m_l$ and  $s_0$ is $x_i$. Therefore, since the receiver's posterior given $m_l$ and  $s_1$ is strictly larger, action $a_{i+1}$ is uniquely played. Therefore, $r_{l,0}<r_{l,1}=u_S(a_{i
+1}).$ Lemma \ref{lem:ind} implies that $r_{h,0}<r_{h,1}=u_S(a_{i
+1})$. Since only actions $a_k,a_{k+1}$ are played with positive probability given $m_{h}$ it follows from our genericity assumption both actions $a_k\text{ and } a_{k+1}$ are played with positive probability given $m_h$ and the high signal $s_1$. That is the receiver's posterior given $m_h$ and the high signal $s_1$ is $x_k$. Therefore, since the receiver's posterior given $m_h$ and  $s_0$ is lower we must have that $r_{h,0}=u_S(a_k)$. Since  $r_{h,0}=r_{l,0}=u_S(a_k)$ and since $r_{l,0}<r_{l,1}=r_{h,1}=u_S(a_{i+1})$. We must have that 
$u_S(a_i)<u_S(a_k)<u_S(a_{i+1})<u_S(a_{k+1})$. Thus $\overline a=(a_i,a_{i+1},a_k,a_{k+1})\in Q$ satisfies the first condition and in addition, the conditions over the utilities hold.     

Assume next that action $a_{i+1}$ is played with  probability zero given $m_l$ and the low signal $s_0$.
In this case, we must have that $r_{l,0}=u_S(a_i).$ 
Since $u_S(a_{i+1})>u_S(a_i)$ we have $r_{l,0}<r_{l,1}$.
Since $r_{l,0}=r_{h,0}$ it follows from the genericity assumption that both actions $a_k$ and $a_{k+1}$ are played with positive probability given $m_{h}$ and $s_0$ and action $a_{k+1}$ is played with probability one given $m_{h}$ and $s_1$. Since by Lemma \ref{lem:ind} $r_{h,0}<r_{h,1}$ 
 it must hold that $u_S(a_k)<u_S(a_{k+1})$. Moreover since 
$r_{l,0}=r_{h,0}$ we must have that
$u_S(a_k)<u_S(a_i)<u_S(a_{k+1})$. Since action $a_{k+1}$ is played with probability one given $m_l$ and  $s_1$ we must have that $r_{h,1}=r_{l,1}=u_S(a_{k+1}).$ This is possible only if $u_S(a_i)<u_S(a_{k+1})<u_S(a_{i+1}).$ Altogether, we have that $u_S(a_k)<u_S(a_i)<u_S(a_{k+1})<u_S(a_{i+1}).$ Thus the equilibrium satisfies condition 3. and it is supported in the interior as desired.  

The fact that if $u_S(a_i)>u_S(a_{i+1})$, then only conditions 2. or 4. are possible, is shown similarly. This completes the proof of the lemma.
\end{proof}

We get the following corollary from Lemma \ref{lem:utility}.
\begin{corollary}\label{cor:au}
 Let $G$ be a private binary information structure for the sender as in Lemma \ref{lem:utility}  that generates the posterior distribution  
 $\beta\delta(a)+(1-\beta)\delta(b)$. Let $0<m_l<\pi<m_h<1$ define a non-trivial cheap-talk equilibrium that is indeterminately supported on four actions $\overline a =(a_i,a_{i+1},a_k,a_{k+1})$ that corresponds to  row $l\in[q]$ in the above matrix $B$. Then the sender equilibrium utility is:
 $\beta B_{l,1}+(1-\beta)B_{l,2}.$
\end{corollary}
\begin{proof}
  Assume that $\overline a$ satisfies condition 1., then it follows from Lemma \ref{lem:utility} that the sender's utility is $\beta u_S(a_k)+(1-\beta)u_S(a_{i+1})$. This by the definition of the matrix $B$ equals 
  $\beta B_{l,1}+(1-\beta)B_{l,2}.$ The case where $\overline a$ satisfies one of the conditions 2-4 follows similarly.
\end{proof}

Before proving Theorem~\ref{prop:quadruple}, we need one additional lemma.
Let $x_1<x_2<x_3$ and $x_4<x_5<x_6$  be six real numbers such that either $x_3<x_4$ or $x_6<x_1$, and such that all of $x_2-x_1$, $x_3-x_2$, $x_5-x_4$, and $x_6-x_5$ are strictly greater than $\eta+\eps$, for some $\eta,\eps>0$. Consider two measures $G_0,G_1\in\Delta([-\eta,\eps])$. Assume that both measures have the same support, and that it contains both $-\eta$ and $\eps$. Assume that $G_1$  first order stochastically dominates $G_0$ with a strict inequality for any point $x\in (-\eta,\eps)$. 
Let $\gamma<\alpha<\delta<\beta$ and let $U$ be a  function such that $U(x)=\alpha$ for $x\in [x_1, x_2)$, $U(x)=\beta$ for $x\in(x_2,x_3]$, $U(x)=\gamma$ for $x\in[x_4,x_5)$, and  $U(x)=\delta$ for $x\in(x_5, x_6]$. At $x=x_2$ (resp., $x=x_5$), the utility $U(x)$ is the range $[\alpha,\beta]$ (resp., $[\gamma,\delta]$). The exact value will be determined by a mixing parameter $r\in [0,1]$, and so denote by $U(x_2)|r = (1-r)\alpha + r\beta$, and by $U(x_5)|r=(1-r)\gamma+r\delta$. For any $x\not\in\{x_2,x_5\}$, let $U(x)|r = U(x)$.
Finally, for every $x\in\real$ and $\omega\in\{0,1\}$ denote by $x+G_\omega$ the shift of $G_\omega$ by $x$. 

We now state and prove Lemma~\ref{lem:two-equalities}.

\begin{lemma}\label{lem:two-equalities}
Under the above conditions there exist $v_b<v_w$ and $r_b, r_w\in [0,1]$ such that
$|v_b-x_2|\leq\max\{\eps,\eta\}$, $|v_w-x_5|\leq\max\{\eps,\eta\}$, and the following two qualities hold:
$$E_{Z\sim v_b+G_0}[U(Z)|r_b]=E_{Z\sim v_w+G_0}[U(Z)|r_w] ~~\mbox{ and }~~ E_{Z\sim v_b+G_1}[U(Z)|r_b]=E_{Z\sim v_w+G_1}E[U(Z)|r_w].$$
\end{lemma}

%
\begin{proof}[Proof of Lemma~\ref{lem:two-equalities}]
For each $p\in[0,1]$, let $f_0(p) = (1-p) \alpha + p \beta$. By Lemma \ref{lem:ivt}, for every $p$ there exist a pair $(v_p, r_p)$ with $v_p\in [x_2-\eps, x_2+\eta]$ and $r_p\in [0,1]$ that satisfy the following:
 Given mixing parameter $r_p$, we have $f_0(p) = E_{Z\sim v_p+G_0}[U(Z)|r_p]$.\footnote{This pair may not be uniquely defined and in this case we can choose the pair $(v_p, r_p)$ arbitrarily.} 
Let $f_1(p) = E_{Z\sim v_p+G_1}[U(Z)|r_p]$. Furthermore, let $s(p)$ satisfy  $f_1(p) = (1-s(p))\gamma + s(p) \delta$. 

In words, $f_0(p)$ is the expected utility of the sender when the receiver plays the higher action $\beta$ with probability $p$, and the lower action $\alpha$ with probability $1-p$. In addition, $f_1(p)$ is the expected utility of the sender when the receiver plays the higher action $\beta$ with probability $s(p)$, and the lower action $\alpha$ with probability $1-s(p)$, where $s(p)$ is chosen such that the belief $v_p$ and mixing probability $r_p$ that lead to $p$ in state $\omega=0$ lead to $s(p)$ in state $\omega=1$. Note that, since $G_0$ and $G_1$ have the same support, and by the assumptions on the distributions, $s(p)\geq p$. In particular, this also implies that $f_1(p)\geq f_0(p)$. Finally, note that both $f_0$ and $f_1$ are continuous with continuous inverses. 

Similarly, for each $p\in[0,1]$, let $h_0(p) = (1-p) \gamma + p \delta$. Each such $p$ defines a pair $(v_p, r_p)$ with $v_p\in [x_5-\eps, x_5+\eta]$ and $r_p\in [0,1]$ that satisfy the following: Given that the receiver plays the action yielding $\delta$ with probability $r_p$ and the action yielding $\gamma$ with probability $1-r_p$ when indifferent, we have $h_0(p) = E_{Z\sim v_p+G_0}[U(Z)|r_p]$.
Let $h_1(p) = E_{Z\sim v_p+G_1}[U(Z)|r_p]$. Furthermore, let $t(p)$ satisfy  $h_1(p) = (1-t(p))\gamma + t(p) \delta$. Note that, as above, $t(p)\geq p$, and so $h_1(p)\geq h_0(p)$. Finally, note that both $h_0$ and $h_1$ are continuous with continuous inverses.

Fix some $c\in [0,1]$ such that $f_1(c) = \delta$.
For $p\in [0, c]$, let
$$H_0(p) = h_0^{-1}(f_0(p)),$$
$$H_1(p) = h_1^{-1}(f_1(p)),$$
and
$$d(p)=H_1(p)- H_0(p).$$

We now show that $d(0)\leq 0$. 
First,  by the definition of $U$ and the fact that the support of $G_\omega$ is contained in $[-\eta,\eps]$ it follows that $f_\omega(0)=\alpha$ for both $\omega\in\{0,1\}$.
In addition, since $x_5-x_4>\eta+\eps,$ it also holds that $h_\omega(0)=\gamma.$ 
Since $\gamma<\alpha<\beta$ there exist $z_\omega\in (0,1)$ such that $h_\omega(z_\omega)=\alpha$ for each $\omega\in\{0,1\}$. We have
$$h_0(z_0)=(1-z_0) \gamma + z_0 \delta = \alpha$$
and
 $$h_1(z_1)=(1-t(z_1)) \gamma + t(z_1) \delta = \alpha,$$
but $t(z_1)\geq z_1$ implies that $z_1\leq z_0$. Finally, observe that $z_1= H_1(0)$ and $z_0 = H_0(0)$, and so $d(0) = z_1-z_0 \leq 0$.
 
Next, we show that $d(c)> 0$. Recall that $c$ satisfies $f_1(c) = \delta$. The set $H_1(c)$ is thus equal to all $y\in[0,1]$ for which $h_1(y)=\delta$. In order for this equality to hold, we must have $H_1(c)=1$.
Next, note that, since $f_0(p)\leq f_1(p)$, we have that $f_0(c)\leq \delta$. This implies that $H_0(c)$ cannot be equal to 1. Thus, $d(c) = 1-H_0(c)>0$.

We have shown that $d(0)\leq 0$ and $d(c)> 0$. By the Intermediate Value Theorem, there exists $b\in [a,c]$ such that $d(b)=0$. This implies that there exists $w\in [0,1]$ such that $f_0(b)=h_0(w)$ and $f_1(b)=h_1(w)$. Noting that $b$ maps to (at least one) pair $(v_b, r_b)$ and $w$ maps to (at least one) pair $(v_w,r_w)$ completes the proof of the lemma.
\end{proof}

We next prove Theorem \ref{prop:quadruple}.
\begin{proof}[\textbf{Proof of Theorem \ref{prop:quadruple}}]
The fact that every $\overline{a}\in Q\cup E$ is robustly supported follows from Lemma \ref{lem:endpointExistence} and Lemma \ref{lem:inermExistence}. It follows from Lemma \ref{lem:utility} that if $F$ is a binary information structure supported on $[\frac{1}{2}-\delta,\frac{1}{2}+\delta]$ and $2\log\Big(\frac{\frac{1}{2}+\delta}{\frac{1}{2}-\delta} \Big)<|J_i|$ for every $1\leq i\leq n$ then any nontrivial cheap talk binary equilibrium with $0<s_l<\pi<s_h<1$ is supported in the interior.

To complete the proof of Theorem \ref{prop:quadruple} we will show that if $F$ supported on $[\frac{1}{2}-\delta,\frac{1}{2}+\delta]$ and $2\log\Big(\frac{\frac{1}{2}+\delta}{\frac{1}{2}-\delta} \Big)<|J_i|$, then any binary equilibrium with $m_l=0<\pi<m_h<1$ or with $0<m_l<\pi<m_h=1$ is supported by an endpoint. 

Assume without loss of generality that $m_l=0<\pi<m_h$. We note that the condition on $F$ implies that at most three actions are played with positive probability. Moreover, the genericity assumption implies that conditional on $s_h$ two consecutive actions $a_i,a_{i+1}$ are played with positive probability. To complete the proof, we need to show that $u_S(a_i)<u_S(a_1)<u_S(a_{i+1})$. Since conditional on state $\omega=0$ both messages $m_l$ and $m_h$ are sent with positive probability, the sender has the same conditional utility from $s_l$ and $s_h$ at state $\omega=0$. This means that we only need to rule out the case where $u_S(a_{i+1})<u_S(a_1)<u_S(a_{i})$. Let $\alpha$ be the probability that action $a_i$ is played given $m_h$ and state $\omega=0$. From the indifference condition, we must have that $\alpha u_S(a_i)+(1-\alpha)u_S(a_{i+1})$. But since $F_1$ first order stochastically dominates $F_0$ the probability that $a_i$ is played given state $\omega=1$ and $m_h$ is strictly smaller than $\alpha$. This means that the sender's conditional utility given the state $\omega=1$ and $m_h$ is strictly smaller than $u_S(a_1)$. This implies that sending the message $m_l$ on state $\omega=1$ yields a profitable deviation for the sender as it guarantees a utility of $u_S(a_1)$. This stands in contradiction to the assumption that $m_l=0<\pi<m_h$ is an equilibrium. 
\end{proof}

\subsection{Proofs of Theorems~\ref{theorem:binary} and~\ref{theorem:general}}

We next turn to prove Theorem \ref{theorem:binary}.
\begin{proof}[\textbf{Proof of Theorem \ref{theorem:binary}}]
Consider a binary information structure $G$ with a support that is contained in $[\frac{1}{2}-\delta,\frac{1}{2}+\delta]$ as in Lemma \ref{lem:utility}.
Assume that $F$ induces the belief distribution $\beta\delta(a)+(1-\beta)\delta(b)$. Then, by Corollary \ref{cor:au} the optimal utility for the sender in a cheap-talk equilibrium that is supported in the interior is $\max_{1\leq l\leq q}\beta B_{l,1}+(1-\beta)B_{l,2}$. 
In addition,  cheap-talk equilibria that are supported in the interior are the only non-trivial equilibria with $0<m_l<\pi<m_h<1$.

Conversely, it follows from Theorem \ref{prop:quadruple} and Lemma \ref{lem:utility} that for all sufficiently small $\delta<0$, any $F\in\mathcal{B}_\delta$, and every $\overline a_l\in Q$ there exists a cheap-talk equilibrium with a utility   
$\beta B_{l,1}+(1-\beta)B_{l,2}$ for the sender. Thus, for every $F$, the maximal sender's utility of a cheap-talk equilibrium is $\beta B_{l,1}+(1-\beta)B_{l,2}.$ Note that for every $\delta>0$ and $\beta\in(0,1)$ there exists a binary signal $F$  with posterior distribution $\beta\delta(a)+(1-\beta)\delta(b)$  supported on $[\frac{1}{2}-\delta,\frac{1}{2}+\delta]$ for some $a<\frac{1}{2}<b$. Therefore, the maximal utility of the sender from a  cheap-talk equilibrium that is supported in the interior can be made arbitrarily close to
$$\min_{0\leq\beta\leq 1}\max_{0\leq l\leq q}\beta B_{l,1}+(1-\beta)B_{l,2}=Val(B).$$

Furthermore, if $a_1$ is achievable, the same logic as in the proof of the Theorem \ref{th:fb} shows that for $F\in\mathcal{B}_\delta$ and sufficiently small $\delta$ there exists an equilibrium $(m_l,m_h)$ with $m_l=0$. Moreover, the sender's utility from such an equilibrium approaches $u_S(a_1)$ as $\delta$ goes to zero. 
A similar conclusion holds for the case $a_n$ is achievable.
Overall we conclude that 
$$\hat v_b(\pi)=\lim_{\delta\to 0}\inf_{F\in\mathcal{B}_\delta}v_F^*(\pi)=\max\{Val(B),v^*_1,v^*_0,u_S(a_j)\},$$
as desired.

\end{proof}

We next prove Theorem \ref{theorem:general}.
\begin{proof}[\textbf{Proof of Theorem \ref{theorem:general}}]
The facts that $\hat v(\pi)\geq v^*_1$ and $\hat v(\pi)\geq v^*_0$ follow as in the proofs of Theorem \ref{th:fb} and Theorem \ref{theorem:binary}. Furthermore, it follows from Theorem \ref{prop:quadruple} that there exists $\delta_0$ such that for every $\delta<\delta_0$ any $\overline a_l\in Q$ and $G\in\mathcal{B}_\delta$ there exists a cheap-talk equilibrium that is supported in the interior.  Theorem \ref{prop:quadruple} implies that that the utility in this equilibrium is at least $\min\{B_{l,1},B_{l,2}\}$. Therefore, $\hat v(\pi)\geq \min\{B_{l,1},B_{l,2}\}$ for any $l\in [q]$. Thus, $\hat v(\pi)\geq V(B)$. This completes the proof of Theorem \ref{theorem:general}.    
\end{proof}


%
%

\end{document}